\documentclass[english,conference]{IEEEtran}
\usepackage[T1]{fontenc}
\usepackage[utf8]{inputenc}
\usepackage{amsthm}
\usepackage{amsmath}
\usepackage{amssymb}
\usepackage{graphicx}

\makeatletter
  \theoremstyle{definition}
  \newtheorem{defn}{\protect\definitionname}
  \theoremstyle{plain}
  \newtheorem{lem}{\protect\lemmaname}
  \theoremstyle{remark}
  \newtheorem{rem}{\protect\remarkname}
\theoremstyle{plain}
\newtheorem{thm}{\protect\theoremname}
  \theoremstyle{plain}
  \newtheorem{fact}{\protect\factname}
  \theoremstyle{plain}
  \newtheorem{cor}{\protect\corollaryname}

\IEEEoverridecommandlockouts
\usepackage{cite}

\ifCLASSINFOpdf
\else
\fi
\usepackage{indentfirst}

\makeatother

\usepackage{babel}

\makeatother

\usepackage{babel}
  \providecommand{\definitionname}{Definition}
  \providecommand{\factname}{Fact}
  \providecommand{\lemmaname}{Lemma}
  \providecommand{\remarkname}{Remark}
\providecommand{\corollaryname}{Corollary}
\providecommand{\theoremname}{Theorem}

\begin{document}

\title{An Upper Bound on Broadcast Subspace Codes}

\author{\IEEEauthorblockN{Yimin Pang}\IEEEauthorblockA{Department of
Information Science\\
 and Electronic Engineering\\
 Zhejiang University\\
 Hangzhou, China, 310027\\
 Email: yimin.pang@zju.edu.cn } \and\IEEEauthorblockN{Thomas Honold}
\IEEEauthorblockA{Department of Information Science\\
 and Electronic Engineering\\
 Zhejiang University\\
 Hangzhou, China, 310027\\
 Email: honold@zju.edu.cn} %
\thanks{Supported by the National Natural Science Foundation of China under
Grant No. 60872063%
}}
\maketitle
\begin{abstract}
Linear operator broadcast channel (LOBC) models the scenario of multi-rate
packet broadcasting over a network, when random network coding is
applied. This paper presents the framework of algebraic coding for
LOBCs and provides a Hamming-like upper bound on (multishot) subspace
codes for LOBCs. 
\end{abstract}

\section{Introduction}

In an acyclic network where random linear network coding\cite{1705002}
is applied, packet transmission in a generation can be regarded as
conveying the subspace spanned by the input vectors \cite{4567581},
and the corresponding equivalent channel model is called \emph{linear
operator channel} (LOC). We will denote the set of all $i$-dimensional
subspaces of $\mathbb{F}_{q}^{m}$ by $\mathcal{P}(\mathbb{F}_{q}^{m},i)$,
which is known as a Grassmannian. The overall set of subspaces of
$\mathbb{F}_{q}^{m}$ is denoted by $\mathcal{P}(\mathbb{F}_{q}^{m})$.
In general, the input and output symbols of a LOC are taken from the
set of all subspaces $\mathcal{P}(\mathbb{F}_{q}^{m})$ of $\mathbb{F}_{q}^{m}$
(referred to as {}``ambient space''). If the input alphabet of a
LOC happens to be $\mathfrak{X}=\mathcal{P}(\mathbb{F}_{q}^{m},l)$
(i.e. all $l$ dimensional subspaces of $\mathbb{F}_{q}^{m}$) we
call it a \textit{constant dimension} LOC. If missing basis vectors
of a subspace constitutes the only possible channel interference,
i.e. no error vectors are inserted into the transmitted subspace,
we say the LOC is \textit{multiplicative}.

An LOC could only be viewed as either unicast channel or constant
rate broadcast channel, however the benefits of network coding are
mainly in packet multicast scenarios. In \cite{Pangb}, a modified
model called \textit{linear operator broadcast channel} (\textit{LOBC})
is proposed to formulate the problem of packet broadcasting over LOCs.
An LOBC is a broadcast channel with subspaces as input/outputs. All
receivers have their possibly different subchannel capacities and
they are able to collect message of their own interests at variable
rates. In \cite{Pangb} the capacity region of a special type of degraded
\emph{constant dimension multiplicative LOBCs (CMLOBCs)} is studied.
CMLOBCs are a generalization of broadcast erasure channels. Although
time sharing is sufficient to achieve the boundary of the capacity
region of degraded broadcast erasure channels, the same conclusion
does not hold for CMLOBCs. This motivates us to study algebraic coding
theory for broadcasting over LOCs. In this paper we set up the framework
of broadcast subspace codes (i.e. a class of multishot subspace codes
with unequal error protection) with separation vectors as performance
parameter and prove an upper bound on the code construction.

Now we briefly summarize previous work on LOC and some related work
about conventional unequal error protection codes. The concept of
LOC was first proposed in \cite{4567581} from a combinatorial viewpoint,
and further studied in \cite{4608992} and many others. Starting with
\cite{5429134}, research on linear operator channels went into the
realm of information theory. In \cite{5429134} Silva et al.\ investigated
the capacity of a random linear network coding channel with matrices
as input/output symbols. Later, by regarding a LOC as a particular
DMC, Uch{ô}a-Filho and N{ó}brega \cite{uchoa2010capacity} studied
the capacity of CMLOCs. Yang et al. \cite{yang2010optimality,yang2010linear}
computed general non-constant multiplicative LOC capacities. In \cite{jafari2009capacity}
the rate region of multiple source access LOCs was investigated. In
\cite{Pangb} packet broadcasting over LOC is presented, and some
initial results on the capacity region of CMLOBCs are proved, while
the capacity region of degraded CMLOBCs (and general LOBCs) remains
unknown.

Rather than using one-shot subspace codes as investigated in \cite{4567581,4608992,5429134},
for achieving the rate region developed in \cite{uchoa2010capacity,yang2010optimality,yang2010linear}
we need multiple uses of the same LOC, which makes it necessary to
construct multishot subspace codes. Some coding strategies can be
found in \cite{5205750,5507933,yang2010coding}. However, constructing
good multishot subspace codes still remains a challenging problem---mainly
due to the apparent lack of a nice group structure ({}``linearity'')
on projective geometries over finite fields, see \cite{etzion-vardy08}.

Unequal error protection (UEP) coding goes back to \cite{1054054},
where Masnick and Wolf suggested techniques to protect code bits in
different levels. Construction and bounds on linear unequal protection
codes can be found in \cite{1056327} and many others. Coding and
modulation issues on UEP are addressed in \cite{243441} and \cite{diggavisuccessive}
at almost the same time.

\section{Coding for Linear Operator Broadcast Channels (LOBCs)}

\subsection{Broadcast Subspace Codes}

Reference \cite{Pangb} considers the case of a multiple user LOC
where a sender communicates with $K$ receivers $u_{1}$, $u_{2}$,...,$u_{K}$
simultaneously. The subchannels from the sender to $u_{k}$, $k=1,2,...,K$,
are linear operator channels with input and output alphabets $\mathfrak{X},\mathfrak{Y}\subseteq\mathcal{P}(\mathbb{F}_{q}^{m})$,
where $m$ and $q$ are fixed. Let $\mathsf{X},\mathsf{Y}_{1},\dots,\mathsf{Y}_{k}$
be the corresponding random variables. The output at every receiver
is taken 
subject to some joint transfer probability distribution $p(Y_{1},Y_{2},\dots,Y_{k}|X)=p_{\mathsf{Y_{1}},\mathsf{Y}_{2},\dots,\mathsf{Y}_{k}|\mathsf{X}}(Y_{1},Y_{2},\dots,Y_{k}|X)=p(\mathsf{Y}_{1}=Y_{1},\mathsf{Y}_{2}=Y_{2},\dots,\mathsf{Y}_{k}=Y_{k}|\mathsf{X}=X)$.
Such a channel is called \emph{Linear Operator Broadcast Channel (LOBC)}.
For simplicity we restrict ourselves to a LOBC with two receivers
and let $\mathfrak{M}_{1}$, $\mathfrak{M}_{2}$ be the alphabets
of private messages for user $u_{1}$ and $u_{2}$, respectively. 
\begin{defn}
A \emph{broadcast (multishot) subspace code} of length $n$ for the
LOBC consists of a set $\mathfrak{C}\subseteq\mathfrak{X}^{n}$ of
codewords and a corresponding encoder/decoder pair. The LOBC encoder
$\gamma\colon\mathfrak{M}_{1}\times\mathfrak{M}_{2}\to\mathfrak{C}$
maps a message pair $(M_{1},M_{2})$ to a codeword $\mathbf{X}=(X_{0},\dots,X_{n-1})\in\mathfrak{C}$
(for every transmission generation). The LOBC decoder $\delta=(\delta_{1},\delta_{2})$
consists of two decoding functions $\delta_{i}\colon\mathfrak{Y}^{n}\to\mathfrak{M}_{i}$
($i=1,2$) and maps the corresponding pair $(\mathbf{Y}_{1},\mathbf{Y}_{2})\in\mathfrak{Y}^{n}\times\mathfrak{Y}^{n}$
of received words to the message pair $(\hat{M}_{1},\hat{M}_{2})=\bigl(\delta_{1}(\mathbf{Y}_{1}),\delta_{2}(\mathbf{Y}_{2})\bigr)$.

The rate pair $(R_{1},R_{2})$, in units of $q$-ary symbols per subspace
transmission, of the broadcast subspace code is defined as

\begin{equation}
R_{1}=\frac{\log_{q}|\mathfrak{M}_{1}|}{n},\quad R_{2}=\frac{\log_{q}|\mathfrak{M}_{2}|}{n}.\label{eq:LOBCRate}
\end{equation}

\end{defn}
As in \cite[Ch.~14.6]{cover1991elements} we can rewrite the encoding
map as 
\[
\gamma\colon(1,2,...,q^{nR_{1}})\times(1,2,...,q^{nR_{2}})\rightarrow\mathfrak{C}
\]
 and associate with the broadcast subspace code the parameters $((q^{nR_{1}},q^{nR_{2}}),n)$.

\subsection{Separation Vector for Broadcast Subspace Codes}

In what follows we view $\mathcal{P}(\mathbb{F}_{q}^{m})^{n}$ as
a metric space relative to the \emph{subspace distance} $d_{S}(\mathbf{X},\mathbf{Y})=\sum_{i=1}^{n}[\dim(X_{i}+Y_{i})-\dim(X_{i}\cap Y_{i})]$,
where $\mathbf{X}=(X_{1},X_{2},...,X_{n})\in\mathcal{P}(\mathbb{F}_{q}^{m})^{n}$,
$\mathbf{Y}=(Y_{1},Y_{2},...,Y_{n})\in\mathcal{P}(\mathbb{F}_{q}^{m})^{n}$. 
\begin{defn}
\label{def:seperationvector} Let $\mathfrak{C}\subseteq\mathcal{P}(\mathbb{F}_{q}^{m})^{n}$
be a multishot subspace code with (bijective) encoding map $\gamma\colon\mathfrak{M}_{1}\times\mathfrak{M}_{2}\to\mathfrak{C}$.
The \textit{separation vector} $\mathbf{s}=(s_{1},s_{2})\in\mathbb{N}^{2}$
($\mathbb{N}$ is the set of positive integers) of $\mathfrak{C}$
with respect to $\gamma$ is defined as 
\begin{align}
s_{1} & =\min_{\substack{M_{1}\neq M_{1}^{'}\\
M_{1},M_{1}^{'}\in\mathfrak{M}_{1}\\
M_{2},M_{2}^{'}\in\mathfrak{M}_{2}
}
}d_{S}\{\gamma(M_{1},M_{2}),\gamma(M_{1}^{'},M_{2}^{'})\}\nonumber \\
s_{2} & =\min_{\substack{M_{2}\neq M_{2}^{'}\\
M_{1},M_{1}^{'}\in\mathfrak{M}_{1}\\
M_{2},M_{2}^{'}\in\mathfrak{M}_{2}
}
}d_{S}\{\gamma(M_{1},M_{2}),\gamma(M_{1}^{'},M_{2}^{'})\}\label{eq:seperationvector}
\end{align}

\end{defn}
The separation vector is the key character to describe the error-correcting
capability of a broadcast subspace code, as indicated in the following 
\begin{lem}
Let an LOBC encoder $\gamma\colon\mathfrak{M}_{1}\times\mathfrak{M}_{2}\to\mathfrak{C}$
have separation vector $\mathbf{s}=(s_{1},s_{2})$ as defined above,
$\mathbf{X}=\gamma(M_{1},M_{2})$ the transmitted codeword, and $\mathbf{Y}\in\mathcal{P}(\mathbb{F}_{q}^{m})^{n}$
the received word. Then we have:

1). A minimum distance decoder at user $u_{1}$ ($u_{2}$) can recover
$M_{1}$ (resp.\ $M_{2}$) from $\mathbf{Y}$ if $2d_{S}(\mathbf{X},\mathbf{Y})<s_{1}$
(resp.\
 $2d_{S}(\mathbf{X},\mathbf{Y})<s_{2}$);

2). The minimum distance of $\mathfrak{C}$ is $\min\{s_{1},s_{2}\}$. 
\end{lem}


Unlike coding for LOCs, it is clear that the performance of broadcast
subspace codes depends on both code and encoder map.

\section{Bounds on Broadcast Subspace Codes}

\subsection{Preparations}
\begin{lem}
\label{lem:HammingUpperBoundlem1} Let $\gamma\colon\mathfrak{M}_{1}\times\mathfrak{M}_{2}\to\mathfrak{C}$
be an LOBC encoder with separation vector $\mathbf{s}=(s_{1},s_{2})$,
where $s_{1}<s_{2}$. Then there exists an auxiliary code $\mathfrak{C}_{\text{aux}}\subseteq\mathfrak{C}$
and two encoding maps $\gamma_{1}:\mathfrak{M}_{1}\rightarrow\mathfrak{C}_{\text{aux}}$
and $\gamma_{2}:\mathfrak{C}_{\text{aux}}\times\mathfrak{M}_{2}\rightarrow\mathfrak{C}$
such that for any $(M_{1},M_{2})\in\mathfrak{M}_{1}\times\mathfrak{M}_{2}$,
$\gamma_{2}(\gamma_{1}(M_{1}),M_{2})=\gamma(M_{1},M_{2})$. 
\end{lem}
\begin{proof} Without loss of generality, we assume $\mathfrak{M}_{1}=\{1,2,....,|\mathfrak{M}_{1}|\}$,
$\mathfrak{M}_{2}=\{1,2,....,|\mathfrak{M}_{2}|\}$ and that the minimum
distance of $\mathfrak{C}$ is attained between two codewords of the
form $\gamma(M_{1},1)$, $\gamma(M_{1}',1)$. Let $\mathfrak{C}_{\text{aux}}=\gamma(\mathfrak{M}_{1},1)$
and $\gamma_{1}(M_{1})=\gamma(M_{1},1)$ for $M_{1}\in\mathfrak{M}_{1}$.
 Since $\gamma_{1}$ is bijective, the map $\gamma_{2}$ is then uniquely
defined by the condition $\gamma_{2}(\gamma_{1}(M_{1}),M_{2})=\gamma(M_{1},M_{2})$.
\end{proof} 
Let $\mathfrak{C}_{\text{cld}}(M_{2})=\gamma(\mathfrak{M}_{1},M_{2})=\gamma_{2}(\mathfrak{C}_{\text{aux}},M_{2})$
be the codeword cluster ({}``cloud'') corresponding to the message
$M_{2}$, so that $\mathfrak{C}=\biguplus\{\mathfrak{C}_{\text{cld}}(M_{2});M_{2}\in\mathfrak{M}_{2}\}$.
For further discussion, we choose a unique representative codeword
$\mathbf{C}_{\text{cldcen}}(M_{2})\in\mathfrak{C}_{\text{cld}}(M_{2})$
to denote the center of $\mathfrak{C}_{\text{cld}}(M_{2})$ and let
$\mathfrak{C}_{\text{cldcen}}=\{\mathbf{C}_{\text{cldcen}}(M_{2})|M_{2}\in\mathfrak{M}_{2}\}$.
Additionally we choose $\mathfrak{C}_{\text{cldcen}}$ to be of minimum
distance $s_{2}$ (which can clearly be done). 
\begin{lem}
\label{lem:HammingUpperBoundlem2} Let $\gamma\colon\mathfrak{M}_{1}\times\mathfrak{M}_{2}\to\mathfrak{C}$
an LOBC encoder with separation vector $\mathbf{s}=(s_{1},s_{2})$,
$s_{1}<s_{2}$, and $\gamma_{1}:\mathfrak{M}_{1}\rightarrow\mathfrak{C}_{aux}$
and $\gamma_{2}:\mathfrak{C}_{aux}\times\mathfrak{M}_{2}\rightarrow\mathfrak{C}$
be defined as above. 
 Then

1) The clouds $\mathfrak{C}_{\text{cld}}(M_{2})$, $M_{2}\in\mathfrak{M}_{2}$
are subspace codes of minimum distance $\geq s_{1}$.

2) The subspace distance between codewords in different clouds $\mathfrak{C}_{\text{cld}}(M_{2})$
and $\mathfrak{C}_{\text{cld}}(M_{2}')$, $M_{2}\neq M_{2}'$, is
at least $s_{2}$. 
\end{lem}

\begin{rem}
To encode messages for a LOBC, we could start by constructing an intermediate
auxiliary code $\mathfrak{C}_{\text{aux}}$ with minimum subspace
distance $s_{1}$, and {}``translate'' the codewords of $\mathfrak{C}_{\text{aux}}$
in some way depending on the message $M_{2}$. For example, in the
case of one-shot codes ($n=1$) we can identify $\mathbb{F}_{q}^{m}$
with the extension field $\mathbb{F}_{q^{m}}$ and use a primitive
element $\alpha$ of $\mathbb{F}_{q^{m}}$ to translate the codewords,
i.e.\ we set $\gamma_{2}(\mathbf{C},j)=\alpha^{j}\mathbf{C}$, where
now $\mathfrak{M}_{2}=\{0,1,\dots,q^{m}-2\}$ (or a suitable subset
thereof). The codeword clouds $\mathfrak{C}_{\text{cld}}(M_{2})$,
$M_{2}\in\mathfrak{M}_{2}$, must then be chosen subject to $d_{S}\bigl(\mathfrak{C}_{\text{cld}}(M_{2}),\mathfrak{C}_{\text{cld}}(M_{2}')\bigr)=\min\bigl\{ d_{S}(\mathbf{C},\mathbf{C}');\mathbf{C}\in\mathfrak{C}_{\text{cld}}(M_{2}),\mathbf{C}'\in\mathfrak{C}_{\text{cld}}(M_{2}')\bigr\}\geq s_{2}$.
In the special case of the {}``Singer cycle construction'' outlined
above this reduces to a condition on the auxiliary code $\mathfrak{C}_{\text{aux}}$,
and the coding procedure reflects the principle of superposition coding
for broadcast channel \cite{720547}.

The relationship between $\mathfrak{C}$, $\mathfrak{C}_{\text{aux}}$
and the cloud centers $\mathbf{C}_{\text{cldcen}}(M_{2})$ is illustrated
in Fig.~\ref{Flo:SphereBounds}. The overall space $\mathcal{P}(\mathbb{F}_{q}^{m})^{n}$
is the ball covered by the largest sphere, the small circles denote
spheres with radius $\left\lfloor \frac{s_{1}-1}{2}\right\rfloor $
in a fixed cluster $\mathfrak{C}_{\text{cld}}(M_{2})$, and the dotted
circles with radius $r_{2}$ denote larger spheres around $\mathbf{C}_{\text{cldcen}}(M_{2})$
matching the error-correcting capabilities of the code $\mathfrak{C}_{\text{cldcen}}$. 
\end{rem}
\begin{figure}
\centering{}\includegraphics[scale=0.75]{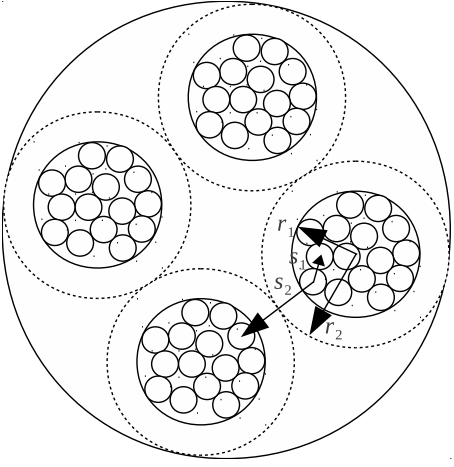}\caption{Sphere Packing of Broadcast Subspace Codes}

\label{Flo:SphereBounds} 
\end{figure}

\subsection{Sphere Packing Bound for General LOBCs}

The ball centered at $\mathbf{V}\in\mathcal{P}(\mathbb{F}_{q}^{m})^{n}$
with radius $r$ is defined as 
\[
\mathcal{B}_{r}(\mathbf{V})=\{\mathbf{U}\in\mathcal{P}(\mathbb{F}_{q}^{m})^{n}|d_{S}(\mathbf{U},\mathbf{V})\leq r\},
\]
 where $d_{S}(\mathbf{U},\mathbf{V})=\sum_{i=1}^{n}d_{S}(U_{i},V_{i})$.
(We omit the symbols $q$ and $m$ in the expression because they
are constant parameters.) The sphere centered at $\mathbf{V}\in\mathcal{P}(\mathbb{F}_{q}^{m})^{n}$
with radius $r$ is defined as 
\[
\mathcal{S}_{r}(\mathbf{V})=\{\mathbf{U}\in\mathcal{P}(\mathbb{F}_{q}^{m})^{n}|d_{S}(\mathbf{U},\mathbf{V})=r\}.
\]
 The volumes of $\mathcal{B}_{r}(\mathbf{V})$ and $\mathcal{S}_{r}(\mathbf{V})$
are defined as 
\[
\text{Vol}(\mathcal{B}_{r}(\mathbf{V}))=|\mathcal{B}_{r}(\mathbf{V})|
\]
 and 
\[
\text{Vol}(\mathcal{S}_{r}(\mathbf{V}))=|\mathcal{S}_{r}(\mathbf{V})|,
\]
 respectively. From \cite{5205750} we know that the volume of $\mathcal{B}_{r}(\mathbf{V})$
only depends on $\mathbf{k}=(\dim V_{1},\dim V_{2},...,\dim V_{n})$
and 
\begin{align*}
\text{Vol}(\mathcal{B}_{r}(\mathbf{V})) & =|\mathcal{B}_{r}(\mathbf{k})|\\
 & =\sum_{\substack{\mathbf{h}\in\{0,1,...,m\}^{n}\\
h_{1}+h_{2}+...+h_{n}\leq r
}
}\prod_{i=1}^{n}\text{Vol}(\mathcal{S}_{h_{i}}(k_{i})),
\end{align*}
 where 
\[
\text{Vol}(\mathcal{S}_{h_{i}}(k_{i}))=\sum_{j=0}^{h_{i}}\dbinom{m-k_{i}}{h_{i}-j}_{q}\dbinom{k_{i}}{j}_{q}q^{j(h_{i}-j)}.
\]
 Since $\mathcal{B}_{r}(\mathbf{V})$ varies in general, we need to
compute the average volume $\text{Vol}^{\text{avg}}[\mathcal{B}_{r}]$
of a ball with radius $r$ in the $\mathcal{P}(\mathbb{F}_{q}^{m})^{n}$
by

\begin{align*}
\text{Vol}^{\text{avg}}[\mathcal{B}_{r}] & =\frac{1}{|\mathcal{P}(\mathbb{F}_{q}^{m})|^{n}}\sum_{\mathbf{V}\in\mathcal{P}(\mathbb{F}_{q}^{m})^{n}}\text{Vol}(\mathcal{B}_{r}(\mathbf{V}))\\
 & =\frac{1}{|\mathcal{P}(\mathbb{F}_{q}^{m})|^{n}}\sum_{\mathbf{k}\in\{0,1,...,m\}^{n}}\dbinom{m}{k_{1}}_{q}\cdots\dbinom{m}{k_{n}}_{q}\text{Vol}(\mathcal{B}_{r}(\mathbf{k}))
\end{align*}
 
\begin{defn}
\label{defn:neighbourhood} The \emph{$r$-neighborhood} of a code
$\mathfrak{S}\subseteq\mathcal{P}(\mathbb{F}_{q}^{m})^{n}$ is defined
as the union of all balls $\mathcal{B}_{r}(\mathbf{V})$ with $\mathbf{V}\in\mathfrak{S}$.
The minimum volume of an $r$-neighborhood of $\mathfrak{S}\subseteq\mathcal{P}(\mathbb{F}_{q}^{m})^{n}$
with $|\mathfrak{S}|=N$ and $d_{S}(\mathfrak{S})\geq d$ is denoted
by $T_{n}(d,N,r)$. \end{defn}
\begin{thm}
For the parameters of an LOBC encoder $\gamma\colon\mathfrak{M}_{1}\times\mathfrak{M}_{2}\to\mathfrak{C}$
with separation vector $\mathbf{s}=(s_{1},s_{2})$, $s_{1}<s_{2}$,
as defined above we have the bound 
 
\begin{equation}
T_{n}\left(s_{1},|\mathfrak{M}_{1}|,\left\lfloor \tfrac{s_{2}-1}{2}\right\rfloor \right)\cdot|\mathfrak{M}_{2}|\leq|\mathcal{P}(\mathbb{F}_{q}^{m})|^{n}.\label{eq:SphereBound}
\end{equation}

\end{thm}
\begin{proof} The clouds $\mathfrak{C}_{\text{cld}}(M_{2})$ have
size $|\mathfrak{M}_{1}|$ and minimum subspace distance at least
$s_{1}$, and hence the volumes of their $r$-neighborhoods are lower-bounded
by $T_{n}(s_{1},|\mathfrak{M}_{1}|,r)$. Further, since clouds corresponding
to different messages $M_{2}$ have distance at least $s_{2}$, their
$\left\lfloor \frac{s_{2}-1}{2}\right\rfloor $-neighborhoods are
pairwise disjoint. This gives \eqref{eq:SphereBound}.\end{proof}



\begin{rem}
The bound \eqref{eq:SphereBound} is similar to the Hamming bound
for linear binary UEP codes derived in \cite{1054054,1056327}. The
numbers $T_{n}(d,N,r)$ are difficult to compute, so that we don't
have an explicit bound as a function of $|\mathfrak{M}_{1}|$ and
$|\mathfrak{M}_{2}|$. (We can use the trivial bound $T_{n}(d,N,r)\geq N\cdot\mathrm{Vol}^{\text{min}}[\mathcal{B}_{(d-1)/2}]$
to convert \eqref{eq:SphereBound} into such an explicit bound, which
is however independent of $s_{2}$.) Using a more elaborate argument,
a lower bound (Varshamov-Gilbert like bound) can also be obtained.
 
\end{rem}

\subsection{Sphere Packing Bound for Constant Dimension LOBCs}

For constant-dimension linear operator broadcast channels, the code
$\mathfrak{C}$ is a subset of the $n$-fold cartesian product of
the Grassmannian $\mathcal{P}(\mathbb{F}_{q}^{m},l)$ consisting of
all $l$-dimensional $\mathbb{F}_{q}$-subspaces of $\mathbb{F}_{q}^{m}$,
where $l\in\{0,1,\dots,m\}$ is some fixed integer; that is $\mathfrak{C}\subseteq\mathcal{P}(\mathbb{F}_{q}^{m},l)^{n}$.
The sphere packing bound on constant dimension LOBCs is presented
in Corollary \ref{cor:SphereBoundCMLOBC}, which depends on the following
fact. 
\begin{fact}
\cite[Lem 4]{4567581} The Gaussian coefficient $\binom{n}{l}_{q}$
satisfies 
\[
q^{l(n-l)}<\binom{n}{l}_{q}<4q^{l(n-l)}
\]
 for $0<l<n$.\end{fact}
\begin{cor}
\label{cor:SphereBoundCMLOBC} Let $\gamma\colon\mathfrak{M}_{1}\times\mathfrak{M}_{2}\to\mathfrak{C}$
be a constant-dimension LOBC encoder with separation vector $\mathbf{s}=(s_{1},s_{2})$,
where $\mathfrak{C}\subseteq\mathcal{P}(\mathbb{F}_{q}^{m},l)^{n}$
and $s_{1}<s_{2}$, as defined above. 
 Then 
 
\begin{equation}
|\mathfrak{M}_{2}|<\frac{4^{n}q^{nl(m-l)}}{T_{n,l}\left(s_{1},|\mathfrak{M}_{1}|,\left\lfloor \tfrac{s_{2}-1}{2}\right\rfloor \right)},\label{eq:CLOBCSphereBound}
\end{equation}
 where $T_{n,l}(d,N,r)$ is the minimum volume of an $r$-neighborhood
in $\mathcal{P}(\mathbb{F}_{q}^{m},l)^{n}$ of a constant-dimension
code $\mathfrak{S}\subseteq\mathcal{P}(\mathbb{F}_{q}^{m},l)^{n}$
of size $N$ and minimum subspace distance $\geq d$. 
\end{cor}

\section{Conclusion}

In this paper, we have converted the problem of superposition coding
for LOBCs into finding multishot subspace codes with required separation
vectors. We have derived sphere packing bounds for general broadcast
subspace codes and for those of constant dimension. The central problem
left for future work will be computing or at least bounding the volume
of an $r$-neighborhood of $\mathfrak{S}\subseteq\mathcal{P}(\mathbb{F}_{q}^{m})^{n}$.
Since we lack a concept of linearity on cartesian products of projective
spaces, multilevel constructions (as suggested in \cite{243441})
of broadcast subspace codes may be a useful alternative to approach
the bound.

\bibliographystyle{IEEEtran} 
 \bibliographystyle{ieeetr}
\bibliography{mathe,MyRef,StdRef,strings}

\begin{thebibliography}{10}
\providecommand{\url}[1]{#1}
\csname url@samestyle\endcsname
\providecommand{\newblock}{\relax}
\providecommand{\bibinfo}[2]{#2}
\providecommand{\BIBentrySTDinterwordspacing}{\spaceskip=0pt\relax}
\providecommand{\BIBentryALTinterwordstretchfactor}{4}
\providecommand{\BIBentryALTinterwordspacing}{\spaceskip=\fontdimen2\font plus
\BIBentryALTinterwordstretchfactor\fontdimen3\font minus
  \fontdimen4\font\relax}
\providecommand{\BIBforeignlanguage}[2]{{%
\expandafter\ifx\csname l@#1\endcsname\relax
\typeout{** WARNING: IEEEtran.bst: No hyphenation pattern has been}%
\typeout{** loaded for the language `#1'. Using the pattern for}%
\typeout{** the default language instead.}%
\else
\language=\csname l@#1\endcsname
\fi
#2}}
\providecommand{\BIBdecl}{\relax}
\BIBdecl

\bibitem{1705002}
T.~Ho, M.~Medard, R.~Koetter, D.~Karger, M.~Effros, J.~Shi, and B.~Leong, ``A
  random linear network coding approach to multicast,'' \emph{Information
  Theory, IEEE Transactions on}, vol.~52, no.~10, pp. 4413 --4430, {O}ct. 2006.

\bibitem{4567581}
R.~Koetter and F.~Kschischang, ``Coding for errors and erasures in random
  network coding,'' \emph{Information Theory, IEEE Transactions on}, vol.~54,
  no.~8, pp. 3579 --3591, {A}ug. 2008.

\bibitem{Pangb}
\BIBentryALTinterwordspacing
Y.~Pang and T.~Honold, ``Towards the capacity region of multiplicative linear
  operator broadcast channels,'' \emph{submitted}, 2010. [Online]. Available:
  \url{http://arxiv.org/abs/1012.5774}
\BIBentrySTDinterwordspacing

\bibitem{4608992}
D.~Silva, F.~Kschischang, and R.~Koetter, ``A rank-metric approach to error
  control in random network coding,'' \emph{Information Theory, IEEE
  Transactions on}, vol.~54, no.~9, pp. 3951 --3967, {S}ept. 2008.

\bibitem{5429134}
D.~Silva, F.~Kschischang, and R.~Kotter, ``Communication over finite-field
  matrix channels,'' \emph{Information Theory, IEEE Transactions on}, vol.~56,
  no.~3, pp. 1296 --1305, {M}ar. 2010.

\bibitem{uchoa2010capacity}
B.~Uch{\^o}a-Filho and R.~N{\'o}brega, ``The capacity of random linear coding
  networks as subspace channels,'' \emph{Arxiv preprint arXiv:1001.1021}, 2010.

\bibitem{yang2010optimality}
S.~Yang, S.~Ho, J.~Meng, and E.~hui Yang, ``Optimality of subspace coding for
  linear operator channels over finite fields,'' in \emph{Proc. IEEE
  Information Theory Workshop}, 2010.

\bibitem{yang2010linear}
S.~Yang, S.~Ho, J.~Meng, and E.~Yang, ``Linear operator channels over finite
  fields,'' \emph{Arxiv preprint arXiv:1002.2293}, 2010.

\bibitem{jafari2009capacity}
M.~Jafari, S.~Mohajer, C.~Fragouli, and S.~Diggavi, ``{On the capacity of
  non-coherent network coding},'' in \emph{Proceedings of the 2009 IEEE
  international conference on Symposium on Information Theory-Volume 1}, 2009,
  pp. 273--277.

\bibitem{5205750}
R.~Nobrega and B.~Uchoa-Filho, ``Multishot codes for network coding: Bounds and
  a multilevel construction,'' in \emph{Information Theory, 2009. ISIT 2009.
  IEEE International Symposium on}, 2009, pp. 428 --432.

\bibitem{5507933}
R.~W. Nobrega and B.~F. Uchoa-Filho, ``Multishot codes for network coding using
  rank-metric codes,'' in \emph{Wireless Network Coding Conference (WiNC), 2010
  IEEE}, 2010, pp. 1 --6.

\bibitem{yang2010coding}
S.~Yang, J.~Meng, and E.~Yang, ``Coding for linear operator channels over
  finite fields,'' in \emph{Information Theory Proceedings (ISIT), 2010 IEEE
  International Symposium on}.\hskip 1em plus 0.5em minus 0.4em\relax IEEE, pp.
  2413--2417.

\bibitem{etzion-vardy08}
T.~Etzion and A.~Vardy, ``Error-correcting codes in projective space,'' in
  \emph{IEEE International Symposium on Information Theory, 2008 (ISIT 2008)},
  Jul. 2008, pp. 871--875.

\bibitem{1054054}
B.~Masnick and J.~Wolf, ``On linear unequal error protection codes,''
  \emph{Information Theory, IEEE Transactions on}, vol.~13, no.~4, pp. 600 --
  607, Oct. 1967.

\bibitem{1056327}
I.~Boyarinov and G.~Katsman, ``Linear unequal error protection codes,''
  \emph{Information Theory, IEEE Transactions on}, vol.~27, no.~2, pp. 168 --
  175, Mar. 1981.

\bibitem{243441}
A.~Calderbank and N.~Seshadri, ``Multilevel codes for unequal error
  protection,'' \emph{Information Theory, IEEE Transactions on}, vol.~39,
  no.~4, pp. 1234 --1248, Jul. 1993.

\bibitem{diggavisuccessive}
S.~Diggavi and D.~Tse, ``{On successive refinement of diversity},'' in
  \emph{Proc. Allerton Conf. Communication, Control, and Computing}, p. 2004.

\bibitem{cover1991elements}
T.~Cover and J.~Thomas, ``{Elements of information theory},'' \emph{Wiley
  Series In Telecommunications}, p. 542, 1991.

\bibitem{720547}
T.~Cover, ``Comments on broadcast channels,'' \emph{Information Theory, IEEE
  Transactions on}, vol.~44, no.~6, pp. 2524 --2530, {O}ct 1998.

\end{thebibliography}
 
\end{document}